\documentclass{llncs}

\usepackage[left=1.3in,right=1.3in,top=1in,bottom=1in]{geometry}
\usepackage{graphicx}

\newcommand{\nd}{neighborhood diversity}

\begin{document}

\title{Algorithmic Meta-theorems for Graphs of Bounded Vertex Cover}
\author{Michael Lampis}
\institute{Computer Science Department,\\ Graduate Center\\ City University of New York\\ \email{mlampis@gc.cuny.edu}}

\maketitle

\begin{abstract}

Possibly the most famous algorithmic meta-theorem is Courcelle's theorem, which
states that all MSO-expressible graph properties are decidable in linear time
for graphs of bounded treewidth. Unfortunately, the running time's dependence
on the MSO formula describing the problem is in general a tower of exponentials
of unbounded height, and there exist lower bounds proving that this cannot be
improved even if we restrict ourselves to deciding FO logic on trees.

In this paper we attempt to circumvent these lower bounds by focusing on a
subclass of bounded treewidth graphs, the graphs of bounded vertex cover. By
using a technique different from the standard decomposition and dynamic
programming technique of treewidth we prove that in this case the running time
implied by Courcelle's theorem can be improved dramatically,  from
non-elementary to doubly and singly exponential for MSO and FO logic
respectively. Our technique relies on a new graph width measure we introduce,
for which we show some additional results that may indicate that it is of
independent interest. We also prove lower bound results which show that our
upper bounds cannot be improved significantly, under widely believed complexity
assumptions.  Our work answers an open problem posed by Michael Fellows.

\end{abstract}

\section{Introduction}

Algorithmic metatheorems are general statements of the form ``\emph{All
problems sharing property P, restricted to a class of inputs I can be solved
efficiently}''. The archetypal, and possibly most celebrated, such metatheorem
is Courcelle's theorem which states that every graph property expressible in
monadic second-order (MSO) logic is decidable in linear time if restricted to
graphs of bounded treewidth \cite{Courcelle90}. Metatheorems have been a
subject of intensive research in the last years producing a wealth of
interesting results.  Some representative examples of metatheorems with a
flavor similar to Courcelle's can be found in the work of Frick and Grohe
\cite{FrickG01}, where it is shown that all properties expressible in first
order (FO) logic are solvable in linear time on planar graphs, and the work of
Dawar et al.  \cite{DawarGKS06}, where it is shown that all FO-definable
optimisation problems admit a PTAS on graphs excluding a fixed minor (see
\cite{Grohe07} and \cite{HlinenyOSG08} for more results on the topic). In all
these works the defining property P for the problems studied is given in terms
of expressibility in a logic language; in many cases metatheorems are stated
with P being some other problem property, for example whether the problem is
closed under the taking of minors. This approach, which is connected with the
famous graph minor project of Robertson and Seymour \cite{RobertsonS83} has
also led to a wealth of significant and practical results, including the so
called bi-dimensionality theory (see \cite{DemaineH08} for an overview and also
the recent results of \cite{BodlaenderFLPST09}).

In this paper we focus on the study of algorithmic metatheorems in the spirit
of Courcelle's theorem, where the class of problems we attack is defined in
terms of expressibility in a logic language. In this research area, many
interesting extensions have followed Courcelle's result: for instance,
Courcelle's theorem has been extended to logics more suitable for the
expression of optimisation problems \cite{ArnborgLS91}. It has also been investigated
whether it's possible to obtain similar results for larger graph classes (see
\cite{CourcelleMR00} for a metatheorem for bounded cliquewidth graphs,
\cite{FominGLS09} for corresponding hardness results and \cite{KreutzerT10} for
hardness results for graphs of small but unbounded treewidth). Finally, lower
bound results have been shown proving that the running times predicted by
Courcelle's theorem can not be improved significantly in general
\cite{FrickG04}. 

This lower bound result is one of the main motivations of this work, because in
some ways it is quite devastating. Though Courcelle's theorem shows that a vast
class of problems is solvable in linear time on graphs of bounded treewidth,
the ``hidden constant'' in this running time, that is, the running time's
dependence on the input's other parameters, which are the graph's treewidth and
the formula describing the problem, is in fact a tower of exponentials.
Unfortunately, in \cite{FrickG04} it is shown that this tower of exponentials
is unavoidable even if we restrict ourselves to deciding FO logic on trees.

From the point of view of meta-theorems the above lead to a rather awkward
situation where a large family of problems can quickly be characterized as
``easy'' on bounded treewidth graphs (by showing the existence of an equivalent
MSO formula), but at the same time we know that at least some of them will in
fact be very hard to solve. Nevertheless, it should be noted that treewidth
research has been an extremely fruitful area and a cornerstone of parameterized
complexity theory, exactly because a large number of generally hard problems is
solvable efficiently (and practically) for graphs of small treewidth (see
\cite{BodlaenderK08} for an excellent survey and the relevant chapters in the
standard parameterized complexity textbooks
\cite{downey1999parameterized,flum2006parameterized,niedermeier2006invitation}).
This apparent disparity between the seemingly prohibitive lower bounds and the
good behavior treewidth exhibits in practice is not due to a huge gap between
the theory and practice\footnote{or more precisely in our case, between
metatheory and theory} of algorithm design for graphs of small treewidth;
rather, as pointed out in Grohe's splenid survey of the field \cite{Grohe07}
the exponential tower in the running time can only be caused by a high number
of quantifier alternations in $\phi$, the formula expressing the problem.
Because many interesting optimization problems can be expressed in MSO logic
with an extremely small number of alternations between existential and
universal quantifiers, they can usually be solved easily. However, this leaves
unanswered the question of what can we do with problems that cannot be
expressed using an extremely small number of quantifier alternations, because
even a modest number of alternations can cause the running time implied by
Courcelle's theorem to sky-rocket.

The above naturally motivate the question of whether a stronger meta-theorem
than Courcelle's can be shown for a sub-class of bounded-treewidth graphs, that
is, a meta-theorem which would imply that all of MSO logic can be solved in
time not only linear in the size of the graph, but also depending reasonably on
the secondary parameters. This question was posed explicitly by Fellows in
\cite{FellowsOpen} for the case of graphs of bounded vertex cover.
Incidentally, this is a class of graphs that has attracted research efforts
again in the past (\cite{FellowsLMRS08}), but in the different direction of
attempting to solve problems which are normally hard for bounded treewidth
graphs and not expressible in MSO logic. The class of bounded vertex cover
graphs is a logical target to attack because the lower bounds we mentioned also
apply to other special cases of bounded treewidth, such as bounded feedback
vertex set (since the lower bound applies even to trees) and bounded pathwidth
(again by \cite{FrickG04}, though not mentioned explicitly). This leaves
bounded vertex cover, which is a further restriction of these as a natural next
candidate.

The main results of this paper show that meta-theorems stronger than
Courcelle's can indeed be shown for this class of graphs. In addition, we show
that our meta-theorems cannot be significantly improved under standard
complexity assumptions. 

In addition to the theoretical interest of these results, there is a potential
for many practical applications here for the many MSO-expressible problems
which require several quantifier alternations to be expressed and are therefore
likely to be hard to solve efficiently for graphs of small treewidth.  Notably,
this class of problems contains for example many two-player games on graphs,
such as Short Generalized Geography and Short Generalized Hex. Such problems
can be expressed in FO logic, a property which generally doesn't seem to
improve things in the case of treewidth but, as we show, improves the running
time exponentially for graphs of small vertex cover.

Specifically, for the class of graphs of vertex cover bounded by $k$ we show
that

\begin{itemize}

\item All graph problems expressible with an FO formula $\phi$ can be solved in
time linear in the graph size and singly exponential in $k$ and $|\phi|$.

\item All graph problems expressible with an MSO formula $\phi$ can be solved
in time linear in the graph size and doubly exponential in $k$ and $|\phi|$.

\item Unless P=NP, there is no algorithm which can decide if an MSO formula
$\phi$ holds in a graph $G$ of vertex cover $k$ in time $f(k,\phi)n^c$, for any
$f(k,\phi)=2^{O(k+|\phi|)}$. Unless $n$-variable 3SAT can be solved in time
$2^{o(n)}$ (that is, unless the exponential time hypothesis fails), then the
same applies for $f(k,\phi)=2^{2^{o(k+|\phi|)}}$.

\item Unless FPT=W[1], there is no algorithm which can decide if an FO formula
$\phi$ with $q$ quantifiers holds in a graph $G$ of vertex cover $k$ in time
$f(k,q)n^c$, for any $f(k,q)=2^{O(k+q)}$.

\end{itemize}

Our upper bounds rely on a technique different from the standard dynamic
programming on decompositions usually associated with treewidth; namely we
exploit an observation that for FO logic two vertices that have the same
neighbors are ``equivalent'' in a sense we will make precise. We state our
results in terms of a new graph ``width'' parameter that captures this graph
property more precisely than bounded vertex cover. We call the new parameter
\nd, and the upper bounds for vertex cover follow by showing that bounded
vertex cover is a special case of bounded \nd.  Our essentially matching lower
bounds on the other hand are shown for vertex cover.  In the last section of
this paper we prove some additional results for \nd, beyond the algorithmic
meta-theorems of the rest of the paper, which we believe indicate that \nd\
might be a graph structure parameter of independent interest and that its
algorithmic and graph-theoretic properties may merit further investigation.

\section{Definitions and Preliminaries}


\subsection{Model Checking, FO and MSO logic}

In this paper we will describe algorithmic meta-theorems, that is, general
methods for solving all problems belonging in a class of problems. However, the
presentation is simplified if one poses this approach as an attack on a simple
problem, the model checking problem. In the model checking problem we are given
a logic formula $\phi$, expressing a graph property, and a graph $G$, and we
must decide if the property described by $\phi$ holds in $G$.  In that case, we
write $G\models \phi$. Clearly, if we can describe an efficient algorithm for
model checking for a specific logic, this will imply the existence of efficient
algorithms for all problems expressible in this logic. Let us now give more
details about the logics we will deal with and the graphs which will be our
input instances.

Our universe of discourse will be labeled, colored graphs. Specifically, we
will assume that the input to our model checking problem consists of a sentence
$\phi$ (in languages we define below) and a graph $G(V,E)$ for which we are
also given a set of labels $L$, each identified with some vertex of $G$ and a
collection of (not necessarily disjoint) subsets of $V$, which we will
informally refer to as color classes.  We will usually denote the set of color
classes of $G$ as $\mathcal{C}=\{C_1,C_2,\ldots,C_c\}$.  The problem we are
truly interested in solving is model checking for unlabeled, uncolored graphs,
which is of course a special case of our definition when $L=\emptyset$ and
$\mathcal{C}=\emptyset$.  The additional generality in our definition is
convenient for two reasons: first, it allows us to more easily express some
problems (for example, checking for a Hamiltonian path with prescribed
endpoints). In addition, in the process of deciding a sentence $\phi$ on a
graph $G$ our algorithm will check through several choices for the vertex and
set variables of $\phi$, which will need to be remembered later by placing a
label on a picked vertex or a color on a picked set of vertices.  Thus, dealing
from the beginning with colored labeled graphs can help to simplify many proofs
by unifying our approach. 

Thus, from now on, we will use the term graph to refer to a labeled colored
graph, that is, a graph $G(V,E)$, a set $L$ and a function $L\to V$, and a set
of colors $\mathcal{C}$ and a function $\mathcal{C}\to 2^V$. We will simply
write $G$ to denote a graph, meaning a graph with this extra information
supplied, unless the labels and colors of $G$ are not immediately clear from
the context. Also, we usually denote $|V|$ by $n$ and for a vertex $v\in V$ we
will write $N(v)$ for the neighborhood of $V$, that is $N(v)=\{u\in V\ |\
(u,v)\in E\}$.

The formulas of FO logic are those which can be constructed inductively using
vertex variables, which we usually denote as $x_i,y_i,\ldots$, vertex labels,
which we usually denote as $l_i$, color classes which we will denote by $C_i$,
the predicates $E(x_i,x_j)$, $x_i\in C_j$, $x_i = x_j$ which can operate on
vertex variables or labels, the usual propositional connectives and the
quantifiers $\exists,\forall$ operating on vertex variables. If a formula
$\phi(x)$ contains an unbound variable $x$ and $l$ is a vertex label we will
denote by $\phi(l)$ the formula obtained by replacing all occurences of $x$
with $l$. A formula $\phi$ is called a sentence if all the variables it
contains are bound by quantifiers.

We define the semantics of FO sentences inductively in the usual way. We will
say that a sentence $\phi$ is true for a labeled colored graph $G$ and write
$G\models \phi$ iff all the vertex and color labels which appear in $\phi$ also
appear in $G$ and 

\begin{itemize}
\item If $\phi=E(l_1,l_2)$ for $l_1,l_2$ two labeled vertices of $G$ which are connected by an edge
\item If $\phi=\ l=l$ for any label $l$
\item If $\phi=l\in C$ for a labeled vertex $l$ which belongs in the color class $C$
\item If $\phi=\phi_1\lor\phi_2$ and $G\models \phi_1$ or $G\models \phi_2$
\item If $\phi=\neg\phi'$ and it is not true that $G\models\phi'$ 
\item If $\phi=\exists x \phi'(x)$ and there exists a vertex of $G$ such that
if $G'$ is the same graph as $G$ with the addition of a new label $l$
identified with that vertex we have $G'\models \phi'(l)$ 
\item If $\phi=\forall x \phi'(x)$ and $G\models \neg \exists x\neg \phi'(x)$

\end{itemize}

Observe that it is not possible for a FO sentence to refer to a specific vertex
of $G$ unless it is labeled.

MSO logic can now be defined in the same way with the addition of set
variables. MSO formulas are constructed in the same way as FO formulas except
that we are now allowed to use set variables $X_i$ and quantify over them, and
the $\in$ predicate can also refer to such variables in addition to color
classes.

For the semantics, we just need to discuss the additional components. In
addition to the rules for FO logic we have that $G\models\phi$ if

\begin{itemize} 

\item  $\phi=\exists X \phi'(X)$ and there exists in $G$ a set of vertices $S$
such that if $G'$ is the same graph as $G$ with the set $S$ added to the set
$\mathcal{C}$ of color classes of $G$ we have $G'\models \phi'(S)$

\end{itemize}

Note that, the MSO logic we have defined here is sometimes also referred to in
the literature as MSO$_1$ logic. This is done to differentiate it from MSO$_2$
logic, where one is also allowed to quantify over sets of edges, not just
vertices. In this paper we focus mostly on MSO$_1$, but we offer some
discussion on MSO$_2$ in Section \ref{sec:nd}.

\subsection{Bounded Vertex Cover and \nd}

Throughout this paper our objective is to prove algorithmic meta-theorems for
graphs of bounded vertex cover, that is, graphs for which there exists a small
set of vertices whose removal also removes all edges. We will usually denote
the size of a graph's vertex cover by $k$. Note that there exist linear-time
FPT algorithms for finding an optimal vertex cover in graphs where $k$ is small
(see e.g. \cite{ChenKX06}). 

Our technique relies on the fact that in a graph of vertex cover $k$, the
vertices outside the vertex cover can be partitioned into at most $2^k$ sets,
such that all the vertices in each set have exactly the same neighbors outside
the set and each set contains no edges inside it. Since we do not make use of
any other special property of graphs of small vertex cover, we are motivated to
define a new graph parameter, called \nd, which intuitively seems to give the
largest graph family to which we can apply our method in a straightforward way. 

\begin{definition}

We will say that two vertices $v,v'$ of a graph $G(V,E)$ have the same type iff
they have the same colors and $N(v)\setminus\{v'\} = N(v') \setminus \{v\}$.

\end{definition}

\begin{definition}

A colored graph $G(V,E)$ has \nd\ at most $w$, if there exists a partition of
$V$ into at most $w$ sets, such that all the vertices in each set have the same
type.

\end{definition}

\begin{lemma} \label{lem:vcnd}

If an uncolored graph has vertex cover at most $k$, then it has \nd\ at most
$2^k+k$.

\end{lemma}

\begin{proof}

Construct $k$ singleton sets, one for each vertex in the vertex cover and at
most $2^k$, one for each subset of vertices of the vertex cover. Place each of
the vertices of the independent set in one of these sets, specifically the one
which corresponds to its neighborhood in the vertex cover. \qed

\end{proof}

In Section \ref{sec:nd} we will show that \nd\ can be computed in
polynomial time and also prove some results which indicate it may be an
interesting parameter in its own right. However, until then our main focus will
be graphs of bounded vertex cover. We will prove all our algorithmic results in
terms of \nd\ and then invoke Lemma \ref{lem:vcnd} to obtain our main
objective. We will usually assume that a partition of the graph into sets with
the same neighbors is given to us, because otherwise one can easily be found in
linear time by using the mentioned linear-time FPT algorithm for vertex cover
and Lemma \ref{lem:vcnd}.

\section{Model checking for FO logic}

In this Section we show how any FO formula can be decided on graphs of bounded
vertex cover using a method that can dramatically improve efficiency, compared
to the standard treewidth-based technique described in Courcelle's theorem. Our
main argument is that for FO logic, two vertices which have the same neighbors
are essentially equivalent. We will prove our results in the more general case
of bounded \nd\ and then show the corresponding result for bounded vertex cover
as a corollary.

Recall that the standard way of deciding an FO formula on a graph is, whenever
we encounter an existential quantifier to try all possible choices of a vertex
for that variable. Because in a graph with small neighborhood diversity most
vertices are equivalent the running time can be drastically reduced.

\begin{lemma} \label{lem:equiv}

Let $G(V,E)$ be a graph and $\phi(x)$ a FO formula with one free variable.  Let
$v,v'\in V$ be two distinct unlabeled vertices of $G$ that have the same type.
Then $G\models \phi(v)$ iff $G\models \phi(v')$.

\end{lemma}

\begin{proof}

Suppose without loss of generality that $\phi(x)$ is in prenex normal form and
has quantifier depth $q$.  We remind the reader that the computation for
$\phi(v)$ can be evaluated by means of a rooted $n$-ary computation tree of
height $q$, where $n=|V|$.  Informally, the children of the root represent the
$n$ possible choices for the first quantified variable of the formula, their
children the choices for the second and so on. Each leaf represents a possible
$q$-tuple of choices for the variables and makes the formula true or false.
Internal nodes compute a value either as the logical disjunction of their
children (for existentially quantified variables) or the logical conjunction
(for universally quantified variables). The value computed at the root is the
truth value of $\phi(v)$.

We will prove the statement by showing a simple correspondence between the
computation trees for $\phi(v)$ and $\phi(v')$. Let $T$ and $T'$ be the two
trees, and label every node of each tree at distance $i$ from the root with a
different tuple of $i$ vertices of $G$ (note that the labels of the tree are
not to be confused with the labels of $G$).  Let $sw_{v,v'}:
\bigcup_{i=1,\ldots,q} V^i \to \bigcup_{i=1,\ldots,q} V^i$ be the ``swap''
function which when given a tuple of vertices of $V$, returns the same tuple
with all occurences of $v$ replaced by $v'$ and vice-versa. As a shorthand,
when $Q$ is a tuple of vertices and $u$ a vertex we will write $(Q,u)$ to mean
the tuple containing all the elements of $Q$ with $u$ added at the end. With
this notation the children of a node with label $Q$ are the nodes with labels
in the set $\{ (Q,u) \ |\ u\in V\}$.

Every leaf in both trees has a $q$-tuple as a label. Let $Q_1$ be such a
$q$-tuple which is the label of a leaf in $T$ and $sw_{v,v'}(Q_1)$ the tuple we
get from $Q_1$ by swapping $v$ with $v'$. The claim is that the leaf of $T$
with label $Q_1$ and the leaf of $T'$ with label $sw_{v,v'}(Q_1)$ evaluate to
the same value. In other words, if we take $\phi(v)$ and replace all quantified
variables with the vertices of $Q_1$ the formula will evaluate to the same
result as when we replace all the quantified variables of $\phi(v')$ with the
vertices of $sw_{v,v'}(Q_1)$. This is true because $\phi$ is a boolean function
of edge, color and equality predicates; color predicates and edge predicates
involving one of $v,v'$ with another vertex are unaffected by swapping $v$ and
$v'$, since these two have the same neighbors and belong in the same color
classes.  Equality predicates are also unaffected since all occurences of $v$
are replaced by $v'$ and vice-versa, thus equality predicates involving these
two and some other vertex will still evaluate to false, while predicates only
involving these two will be unaffected because equality is symmetric. Finally,
edge predicates involving only $v$ and $v'$ are unaffected since $E()$ is
symmetric. Thus, we have established a one-to-one correspondence between the
leaves of $T$ and $T'$ via the function $sw_{v,v'}$, preserving truth values.

Now, we need to establish a correspondence between the internal nodes, again
via $sw_{v,v'}$. Consider a node of $T$ with label $Q_1$ and the node of $T'$
with label $sw_{v,v'}(Q_1)$. The children of the former have labels in the set
$C_1=\{ (Q_1,u)\ |\ u\in V\}$. The children of the latter have labels in
$C_2=\{ (sw_{v,v'}(Q_1),u)\ |\ u\in V\}$. It is not hard to see that $C_2 = \{
sw_{v,v'}(Q)\ |\ Q\in C_1\}$, or in other words, the correspondence between
nodes is transferred up the levels of the trees.

The only remaining part is to establish that if two nodes in $T$ and $T'$ have
labels corresponding via $sw_{v,v'}$, then they compute the same value. We
already established this for the leaves. For internal nodes, this follows from
the fact that the sets of children of two corresponding nodes are also in
one-to-one correspondence via $sw_{v,v'}$ and that the nodes are both of the
same type (existential or universal) since only nodes at the same level can be
corresponding. Thus, by an inductive argument, all the children of the roots of
the two trees compute the same values and therefore $\phi(v)$ and $\phi(v')$
are equivalent. \qed

\end{proof}

\begin{theorem} \label{thm:fo}

Let $\phi$ be a FO sentence of quantifier depth $q$. Let $G(V,E)$ be a labeled
colored graph with \nd\  at most $w$ and $l$ labeled vertices.  Then, there is
an algorithm that decides if $G\models \phi$ in time $O((w+l+q)^q\cdot
|\phi|)$. 

\end{theorem}

\begin{proof}

We will rely heavily on Lemma~\ref{lem:equiv} and describe an inductive
argument. If $q=0$ the problem is of course trivial so assume that $q>0$ and
the theorem holds for sentences of depth at most $q-1$.  Also, assume wlog that
$\phi$ is in prenex normal form and furthermore, that $\phi = \exists x
\psi(x)$, since the universal case can be easily decided if we solve the
existential case, by deciding on the negation of $\phi$.

Suppose that $V$ can be partitioned into $V_1,V_2,\ldots V_w$ as required by
the definition of \nd. Now, by Lemma~\ref{lem:equiv} if $v,v'\in V_i$ for some
$i$, and neither of the two is labeled then $G\models \psi(v)$ iff $G\models
\psi(v')$. Thus, we need to model check at most $(w+l)$ sentences of $q-1$
quantifiers to decide $\phi$: we try replacing $x$ with each of the $l$ labeled
vertices or with one arbitrarily chosen representative from each $V_i$. In the
process we introduce a new label. Repeating this process constructs a
computation tree with at most $\prod_{i=0}^{q-1} (w+l+i) =
O\left((w+l+q)^q\right)$ leaves.  The result of the computation tree can be
evaluated in time linear in its size.  \qed

\end{proof}

\begin{corollary}

There exists an algorithm which, given a FO sentence $\phi$ with $q$ variables
and an uncolored, unlabeled graph $G$ with vertex cover at most $k$, decides if
$G\models \phi$ in time $2^{O(kq+q\log q)}$.

\end{corollary}

Thus, the running time is (only) singly exponential in the parameters, while a
straightforward observation that bounded vertex cover graphs have bounded
treewidth and an application of Courcelle's theorem would in general have a
non-elementary running time. Of course, a natural question to ask now is
whether it is possible to do even better, perhaps making the exponent linear in
the parameter, which is $(k+q)$.  As we will see later on, this is not possible
if we accept some standard complexity assumptions.

\section{Model checking for MSO logic}

In this section we will prove a meta-theorem for MSO logic. It's worth noting
again that the the logic we refer to as MSO is also sometimes called MSO$_1$
logic in the literature, because we only allow quantifications over vertex
sets, as opposed to MSO$_2$, where quantification over edge sets is also
allowed. Courcelle's theorem for treewidth also covers MSO$_2$ logic, which we
don't touch on in this Section, but we give some relevant discussion in
Section~\ref{sec:nd}.

First, let us show a helpful extension of the results of the previous Section.
From the following Lemma it follows naturally that the model checking problem
for MSO logic on bounded vertex cover graphs is in XP, that is, solvable in
polynomial time for constant $\phi$ and $k$, but our objective later on will be
to do better. We will again use the concept of vertex types; recall that two
vertices have the same type if they have the same neighbors and the same
colors.

\begin{lemma} \label{lem:msoxp}

Let $\phi(X)$ be an MSO formula with a free set variable $X$.  Let $G$ be
a graph and $S_1,S_2$ two sets of vertices of $G$ such that all vertices of
$(S_1\setminus S_2)\cup (S_2\setminus S_1)$ are unlabeled and have the same
type and furthermore $|S_1\setminus S_2|=|S_2\setminus S_1|$.  Then $G\models
\phi(S_1)$ iff $G\models \phi(S_2)$.

\end{lemma}

\begin{proof}

The proof follows ideas similar to those of Lemma~\ref{lem:equiv}. Suppose that
$\phi(X)$ has $q$ quantifiers in total, then it is possible to decide if
$G\models \phi(S)$ using a computation tree such that for each quantified
variable we have nodes in the tree with $n$ children and for each quantified
set variable we have nodes with $2^n$ children, with each child corresponding
to a possible choice for that variable. Again, we can label each node of the
tree with a tuple of at most $q$ elements, but now the elements can be either
individual vertices or sets of vertices.

Observe that it suffices to prove the claim when $|S_1\setminus S_2| =
|S_2\setminus S_1| = 1$, because then we can apply the claim repeatedly to
transform $S_1$ to $S_2$ by exchanging the different vertices one by one. So,
suppose that $S_1\setminus S_2 = v$ and $S_2 \setminus S_1 = v'$, and $v$ and
$v'$ have the same type.

Now, the $sw_{v,v'}$ function of Lemma~\ref{lem:equiv} can be extended to act
on sets of vertices in a straightforward way. Consider the computation trees
for $\phi(S_1)$ and $\phi(S_2)$. Once again we must show that $sw_{v,v'}$ is a
one-to-one correspondence between the leaves of the two trees that preserves
truth values. For edge and equality predicates we can use the same arguments as
in Lemma~\ref{lem:equiv}, so the only difference can be with predicates of the
form $x\in X$. However, it is not hard to see that the truth values of these is
not affected when $x\neq v,v'$ and also when $X$ is one of the supplied colors
of the graph, since $v,v'$ have the same colors. Finally, the truth value is
also unaffected if $X$ is a variable set, since $sw_{v,v'}$ is applied both to
vertex and set variables. Now, the correspondence is lifted up the levels of
the tree using similar arguments and this completes the proof. \qed

\end{proof}

\begin{lemma} \label{lem:msofpt}

Let $\phi(X)$ be an MSO formula with one free set variable $X$, $q_V$
quantified vertex variables and $q_S$ quantified set variables. Let $G$ be a
graph and $S_1,S_2$ two sets of vertices of $G$ such that all vertices of
$(S_1\setminus S_2)\cup (S_2\setminus S_1)$ are unlabeled and belong in the
same type $T$. Suppose that both $|S_1\cap T|$ and $|S_2\cap T|$ fall in the
interval $[2^{q_S}q_V,|T|-2^{q_S}q_V-1]$. Then $G\models \phi(S_1)$ iff
$G\models \phi(S_2)$.

\end{lemma}

\begin{proof}

We are dealing with the case where two sets are different, but their different
elements are all of the same type. To give some intuition, in the base case of
$q_S=0$ for this particular type both sets have the property that the sets
themselves and their complements have at least $q_V$ vertices of the type. This
will prove important because $\phi(X)$ will be a FO sentence after we decide on
a set for $X$ and as we will see an FO sentence cannot distinguish between two
different large enough sets (informally, we could say that an FO sentence with
$q$ quantifiers can only count up to $q$). We will show how to extend this to
general $q_S$ by shrinking the interval of sizes where we claim that sets are
equivalent, because every set variable $X_i$ essentially doubles the amount we
can count, by partitioning vertices into two sets, those in $X_i$ and those in
its complement.

First, assume without loss of generality that $|S_1|\le |S_2|$. Now because of
Lemma~\ref{lem:msoxp} we can further assume without loss of generality that
$S_1\subseteq S_2$, because there exists a set $S_2'$ of the same size as $S_2$
such that $S_1\subseteq S_2'$ and $G\models \phi(S_2)$ iff $G\models
\phi(S_2')$.  Furthermore, we may focus on the case where $S_2=S_1\cup\{u\}$
for some vertex $u\not\in S_1$, because if we prove the statement for sets
whose sizes only differ by 1, then we can apply it repeatedly to get the
statement for sets which have a larger difference.

We will now rely on Lemma~\ref{lem:msoxp} to construct an XP algorithm for
deciding $\phi(S_1)$ and $\phi(S_2)$. The trivial algorithm we have already
discussed would consider $2^n$ sets every time a set variable has to be
assigned a value and $n$ vertices every time a vertex variable has to be
assigned a value.  However, because of Lemma~\ref{lem:msoxp} we can consider
only $O(2^ln^w)$ different assignments for a set variable. This is because the
equivalence between different sets of the same size established allows us to
sample one set for each combination of sizes that the set will have with each
of the $w$ types (the $2^l$ factor comes from the fact that labeled vetices are
``special'' and we have to decide for each one individually). Note though that
deciding on an assignment of a set can in the worst case double $w$, since we
are adding a new color to the graph representing the set. Thus, for the next
set we would have to consider $O(2^ln^{2w})$ choices and so on.  Furthermore,
from the proof of Lemma~\ref{lem:msoxp} it is straightforward to derive a
slightly stronger version of Lemma~\ref{lem:equiv} which holds for MSO
sentences. Using this we conclude that we need to check through $w+l$ samples
when we are deciding on a vertex variable and this introduces a new label.

Suppose that we use the algorithm sketched above to decide $\phi(S_1)$ and
$\phi(S_2)$. The crucial point now is that this algorithm has a lot of freedom
in picking the sample sets and vertices it considers. In particular, when
assigning value to a vertex variable the algorithm can always avoid the vertex
$u$ if there are still other vertices of the same type. It is not hard to see
that if the algorithm never assigns $u$ to any vertex variable when deciding
$\phi(S_1)$ and $\phi(S_2)$ the result will necessarily be the same for both
sentences. So we need to argue why the algorithm can always avoid using $u$.

To achieve this we can exploit the freedom the algorithm has when picking sets.
Every time the algorithm picks a set to be considered the set of vertices of
the same type as $u$ is partitioned into two sets. Because it does not matter
which vertices are included in a set and only the size of the set's partition
with a type matters, we can make sure that $u$ is always placed in the larger
of the two new types by exchanging with another vertex appropriately. Because
of the restriction on the sizes of $S_1$ and $S_2$ we know that initially $u$
belongs in a type shared by at least $2^{q_S}q_V$ other vertices.  It is not
hard to see that this invariant is maintained by the algorithm when picking a
set if we place $u$ in the larger of the two new types when picking a set and
we pick a different sample from its type when we pick a vertex. Thus, we have
established that there exists an algorithm that decides $\phi(S_1)$ and
$\phi(S_2)$ without ever assigning $u$ to a vertex variable, which means that
the algorithm must decide the same value for both sentences. \qed

\end{proof}

\begin{theorem} \label{thm:msofpt}

There exists an algorithm which, given a graph $G$ with $l$ labels, \nd\ at
most $w$ and an MSO formula $\phi$ with at most $q_S$ set variables and $q_V$
vertex variables, decides if $G\models \phi$ in time $2^{ O\left(2^{q_S} (w+l)
q_S^2 q_V \log q_V \right) }\cdot |\phi|$.

\end{theorem}

\begin{proof}

Our algorithm now will rely heavily on Lemma~\ref{lem:msofpt}. When picking an
assignment for a set variable, for each of the $w$ types of vertices we need to
decide on the size of its intersection with the set. Because of
Lemma~\ref{lem:msofpt} we can limit ourselves to considering $2^{q_S+1}q_V$
different sizes for the first set, which gives $(2^{q_S+1}q_V)^w$ choices for
the first set variable. However, because every time we decide on a set we start
working on a graph with one more color, the number of vertex types may at most
double. From these we can derive an easy upper bound on the number of
alternatives we will consider for each set variable as $2^{2^{q_S}w(q_S+1+\log
q_V)}$. Since we have $q_S$ set variables in total this gives
$2^{q_S2^{q_S}w(q_S + 1 +\log q_V)}$. For each vertex variable we have to
consider at most $2^{q_S}w+l+q_V$ alternatives, so for all $q_V$ variables at
most $(2^{q_S}w+l+q_V)^{q_V}$. The product of these two upper bounds is an
upper bound on the total number of alternatives our algorithm will consider,
giving the promised running time. \qed

\end{proof}

\begin{corollary}

There exists an algorithm which, given an MSO sentence $\phi$ with $q$
variables and an uncolored, unlabeled graph $G$ with vertex cover at most $k$,
decides if $G\models \phi$ in time $2^{2^{O(k+q)}}$.

\end{corollary}

Again, this gives a dramatic improvement compared to Courcelle's theorem,
though exponentially worse than the case of FO logic. This is an interesting
point to consider because for treewidth there does not seem to be any major
difference between the complexities of model checking FO and MSO logic.

The natural question to ask here is once again, can we do significantly better?
For example, perhaps the most natural question to ask is, is it possible to
solve this problem in $2^{2^{o(k+q)}}$?  As we will see later on, the answer is
no, if we accept some standard complexity assumptions.

\section{Lower Bounds}

In this Section we will prove some lower bound results for the model checking
problems we are dealing with. Our proofs rely on a construction which reduces
SAT to a model checking problem on a graph with small vertex cover. 

For simplicity, we first present our construction for directed graphs. Even
though we have not talked about directed graphs thus far, it is quite immediate
to extend FO and MSO logic to express digraph properties; we just need to
replace the $E()$ predicate, with a non-symmetric predicate for the digraph's
arcs. To avoid confusion we use $D(x,y)$ to denote the predicate which is true
if a digraph has an arc from $x$ to $y$. It is not hard to see that the results
of Theorems~\ref{thm:fo} and \ref{thm:msofpt} easily carry over in this setting
with little modification; we just need to take into account that a digraph of
vertex cover $k$ has $4^k$, rather than $2^k$ categories of vertices. After we
describe our construction for labeled, colored digraphs, we will sketch how it
can be extended to unlabeled, uncolored graphs.

Given a propositional 3-CNF formula $\phi_p$ with $n$ variables and $m$
clauses, we want to construct a digraph $G$ that encodes its structure, while
having a small vertex cover. The main problem is encoding numbers up to $n$
with graphs of small vertex cover. Here, we extend the basic idea of
\cite{FrickG04} where numbers are encoded into directed trees of very small
height, but rather than using a tree we construct a DAG.

We define the graph $N(i)$ inductively:

\begin{itemize}

\item $N(0)$ is just one vertex

\item For $i>0$, $N(i)$ is the graph we obtain from $N(i-1)$ by adding a new
vertex. Let $i_j$ denote the $j$-th bit of the binary representation of $i$,
with the least significant bit numbered 0. Let $H=\{ j\ |\ i_j=1 \}$. Then for
all elements $j\in H$ we add an arc from the new vertex to the vertex which was
first added in the graph $N(j)$.

\end{itemize}

In our construction we will use 6 copies of $N(\log n)$ and refer to them as
$N_i$, $1\le i \le 6$. We will also informally assume in our argument a
numbering for the vertices of each $N_i$, from $0$ to $\log n$, in the order in
which they were added in the inductive construction we described. We will
informally say that each vertex corresponds to a number. (Note that this
numbering is only used in our arguments, we are not assuming that these
vertices are labeled).

The digraph will now consist of the six copies of $N(\log n)$ we mentioned and
two additional sets of vertices:

\begin{itemize}

\item The set $V_1=\{v_1,\ldots,v_n\}$ whose vertices correspond to variables.
For each $v_i\in V_1$ we add an arc to vertex $j$ of the set $N_1$ iff the
$j$-th bit of the binary representation of $i$ is 1.

\item The set $M=\{u_1,\ldots,u_m\}$ whose vertices correspond to clauses. For
the vertex $u_i$ which corresponds to a clause with three literals,
$l_1,l_2,l_3$. If $l_1$ is a positive literal, we add arcs from $u_i$ to
vertices of $N_1$ which correspond to bits of the binary representation of the
variable of $l_1$. If it is a negative literal, we add the same arcs but to
vertices of $N_2$. Similarly, if $l_2$ is positive, we add arcs from $u_i$ to
vertices in $N_3$, otherwise to $N_4$, and for $l_3$ to $N_5$ and $N_6$.

\end{itemize}

To complete the construction of the digraph $G$ we need just to specify the
labels and colors used. The label set will be empty, while the color set will
simply be $\mathcal{C}=\{N_1,N_2,N_3,N_4,N_5,N_6,V_1,M\}$.

We now need to define a formula $\phi$, such that $G\models \phi$ iff $\phi_p$
is satisfiable. First, we need a way to compare the numbers represented by
different vertices of $G$. We inductively define a formula $eq_h(x,y,C_1,C_2)$.
Informally, its meaning will be to compare the numbers represented by two
vertices $x$ and $y$ by checking out-neighbors of $x$ in color class $C_1$ and
out-neighbors of $y$ in color class $C_2$. The main concept of $eq_h$ is
similar as that of the construction in Chapter 10.3 of
\cite{flum2006parameterized}, but in our case it is necessary to complicate the
construction by adding the color classes because this will allow us to
independently check the number represented by each of the three literals in a
clause. First, we set $eq_0(x,y,C_1,C_2) = \forall z z$, that is we set $eq_0$
to be trivially true. Now assuming that $eq_h$ is defined as a first attempt we
could set that

\begin{eqnarray*} 
eq_{h+1}(x,y,C_1,C_2) &=& \forall w \Big( (D(x,w)\land w\in
C_1) \to \\
                      & &            \exists z( D(y,z) \land  y\in C_2 \land eq_h(w,z,C_1,C_2)) \Big) \land \\
			&& \forall z' \Big( (D(y,z')\land z'\in
C_2) \to \\
                      & &            \exists w'( D(x,w') \land  w'\in C_1 \land eq_h(w',z',C_1,C_2)) \Big)  \\
\end{eqnarray*}

However, this would make our formula too large, because then $eq_h$ would grow
exponentially in $h$. We will follow the trick presented in
\cite{flum2006parameterized}, where it is observed that the above formula is
equivalent to 

\begin{eqnarray*} 
eq_{h+1}(x,y,C_1,C_2) =& \\
	& &\Big( (\exists w D(x,w) \land w\in C_1) \leftrightarrow (\exists z D(y,z) \land z\in C_2) \Big) \land\\
	& &\forall w ( (D(x,w)\land w\in C_1) \to \\
	& &\exists z ( (D(y,z)\land z\in C_2) \land \\
	& &\forall z'( (D(y,z')\land z'\in C_2) \to \\
	& &\exists w'( (D(x,w')\land w'\in C_1) \land \\
	& &eq_h(w,z,C_1,C_2) \land eq_h(w',z',C_1,C_2) ))))
\end{eqnarray*}

Though this definition would still make $eq_h$ have size exponential in $h$, we
can now see that $eq_h(w,z,C_1,C_2) \land eq_h(w',z',C_1,C_2)$ is equivalent to 

$$ \forall u \forall v \Big( ((u=w\land v=z) \lor (u=w'\land v=z'))\to
eq_h(u,v,C_1,C_2)\Big) $$

Using this last trick, it is not hard to show with a simple induction that the
size of $eq_h$ is $O(h)$.

Let $tow(h)$ be the function inductively defined as $tow(0)=0$ and
$tow(h+1)=2^{tow(h)}$. From now on we will use $h_n$ to denote the minimum $h$
such that $tow(h)\ge n$ (that is, $h_n= \log^*n$).

The formula $\phi$ we construct will be

\begin{eqnarray*}
 \exists S & \forall x ( x\in M \to \\
           & \exists y ( y\in V_1 \land \\
	   & &\Big( (y\in S \land ( \bigvee_{i\in\{1,3,5\}} eq_{h_n}(x,y,N_i,N_1))) \lor \\
	   &       &( y\not\in S \land (\bigvee_{i\in\{2,4,6\}} eq_{h_n}(x,y,N_i,N_1))) \Big) ))
\end{eqnarray*}

We can now establish the following facts:

\begin{lemma}

$G\models \phi$ iff $\phi_p$ is satisfiable. Furthermore, $\phi$ has size
$O(\log^*n)$, using 1 set quantifier and $O(\log^*n)$ vertex quantifiers and
$G$ has a vertex cover of size $O(\log n)$.

\end{lemma}

\begin{proof}

The only non-trivial part to verify is that $eq_{h_n}(x,y,N_i,N_1)$ works as
expected, that is, it will be true iff $y$ does indeed
correspond to a variable which appears in the clause which corresponds to $x$.
To prove this it suffices to prove that $eq_{(h_n-1)}(x,y,N_i,N_j)$ works
correctly for any two vertices $x\in N_i$ and $y\in N_j$, meaning that it is
true iff $x$ and $y$ correspond to the same number. We will show this by induction
on $h$.  Specifically, we will show that for all $h$, $eq_h$ works correctly
for the first $tow(h)$ vertices of the sets $N_i$.  This will imply that
$eq_{(h_n-1)}$ works correctly for the first $tow(h_n-1)\ge \log n$ vertices of
the sets $N_i$, that is, for the whole sets.

The base case is that $eq_0(x,y,N_i,N_j)$ works correctly for the vertices of
$N_i$ and $N_j$ corresponding to 0. In this case $eq_0$ is of course always
true, which makes the base case trivial.

Suppose that we have established the inductive hypothesis up to some $h$, that
is, we know that $eq_h(x,y,N_i,N_j)$ is true iff $x$ and $y$ correspond to the
same number, assuming that this number is at most $tow(h)$. It is not hard to
see that using this we can establish the correctness of $eq_{h+1}$ for all
vertices up to $2^{tow(h)}$, because these vertices only have out-neighbors
corresponding to numbers up to $tow(h)$.  \qed 

\end{proof}

Let us now describe how our construction can be extended to undirected graphs.

\begin{lemma}

Let $G$ and $\phi$ be as in the construction above. Then there exists an
uncolored, unlabeled graph $G'$ and an MSO sentence $\phi'$ such that $G\models
\phi$ iff $G'\models \phi'$. Furthermore, the vertex cover of $G'$ is $O(\log
n)$ and $\phi'$ has $O((\log^*n)^2)$ vertex variables and one set variable.

\end{lemma}

\begin{proof}

First, let us describe how to make the graph undirected. Observe that $G(V,A)$
is a DAG. We define for every vertex $v$ of $G$ the value $l(v) = 1+
\max_{(v,u)\in A} l(u)$ if $v$ is not a sink and $l(v)=1$ if it is.
Informally, $l(v)$ is the order of the longest path that can be constructed
from $v$ to a sink.  Note that the maximum $l(v)$ in $G$ is $\log^*n$.

Add to $G$ a directed path on $\log^*n$ vertices and number the vertices of the
path $1,2,\ldots,\log^* n$, starting from the sink. Now, from every vertex $u$
of $G$ add an arc to the vertex $l(u)$ of the path. Add a new color class $P$
to the graph, which includes the vertices of the path. Also add a label, $l_s$
identified with vertex $1$ of the path.

Now, we can remove the directions of the arcs of $G$ to obtain an  undirected
graph. In order to retain the proper meaning of $\phi$ in the new graph we must
replace all $D(x,y)$ predicates with $E(x,y) \land \psi(x,y)$, where
$\psi(x,y)$ will be a formula whose informal meaning is that $l(x)>l(y)$.  This
can be expressed using the path we added. 

First, we construct the formula 

\begin{eqnarray*} 
P(x,S) &=& (x\in P) \land \\
       & & \Big( \exists S\ ( \forall y (y\in S \to y\in P)) \land (x\in S) \land (l_s \in S) \\
       & &  \land (\forall y\  y\in S \to  \\
       & &\ \ (\exists z_1 \exists z_2\  (z_1\in S) \land (z_2\in S) \land E(y,z_1) \land E(y,z_2) \land
z_1\neq z_2))\Big)
\end{eqnarray*}

Informally, this formula is true iff $x$ is a vertex of $P$ and $S$ a set of
vertices of $P$ that induce a path from $x$ to $l_s$. Note that, we could also
express $P(x,S)$ with FO logic, if we use an extra $O(\log^* n)$ variables,
since the size of $S$ is upper-bounded by $O(\log^* n)$. So using this bound we
can simply consider $\exists S$ to be shorthand for $O(\log^* n)$ existential
quantifiers. In the remainder we will use the set notation, with the
understanding that it can be thus eliminated if we so desire.

Now, we are ready to define $\psi(x,y)$

\begin{eqnarray*}
\psi(x,y)&=& \exists x'\ \exists y'\ E(x,x') \land E(y,y') \land (\exists S_x
\exists S_y \\
         & &\ \ \ P(x',S_x) \land P(y',S_y) \land (\forall z\ z\in S_y\to z\in
S_x))
\end{eqnarray*}

The intuition behind our construction is that the direction of the arcs of a
DAG can be recovered from the underlying undirected graph if we remember for
every vertex its maximum distance from a sink. We achieve this by connecting
every vertex to an appropriate vertex of an auxilliary path $P$, in a way
``projecting'' paths from the DAG to $P$. Now comparison between two paths can
be performed in our logic simply by checking the projected paths on $P$, since
one must be a subset of the other. 

Thus, we have constructed an undirected graph, which uses one label and a
constant number of colors, and a formula which we can model check on this new
graph. Note that the new formula is not much larger than the old one: we have
replaced all of the $O(\log^* n)$ occurences of the $D()$ predicate with a
formula of constant size for MSO logic, or size $O(\log^* n)$ for FO logic, if
we replace the sets as described previously.

Now, the last step is showing how to get rid of colors and labels. First,
eliminating colors is straightforward if we are willing to add a few additional
labels to our graph. Add one labelled vertex for each color class and connect
it with all the vertices belonging in that class. Now, the $\in$ predicate can
be replaced with a check for a connection to the labelled vertex of the color
class. 

Finally, to eliminate labels, it suffices to notice that our graph has no
leaves. Thus, attaching a leaf to a vertex is enough to make it special, and
checking if a vertex has a leaf attached to it can be performed by a constant
size FO formula. Because we need $O(1)$ (specifically, 10) labels, we attach a
different number of leaves to each vertex which would be labelled. We can now
add $O(1)$ variables to our formula and force each to be identified with each
vertex we need labeled, without increasing the size of the formula by more than
a constant.

In the end we have a graph with vertex cover still $O(\log n)$, and a formula
with 1 set variable and $O((\log^* n)^2)$ vertex variables. Our graph is
unlabeled and uncolored. \qed

\end{proof}

\begin{theorem}

Let $\phi$ be a MSO formula with $q_v$ vertex quantifiers, $q_S$ set
quantifiers and $G$ a graph with vertex cover $k$. Then, unless P=NP, there is
no algorithm which decides if $G\models \phi$ in time $O(2^{O(k+q_S+q_V)}\cdot
poly(n))$. Unless NP$\subseteq$ DTIME($n^{poly\log(n)}$), there is no algorithm
which decides if $G\models \phi$ in time $O(2^{poly(k+q_S+q_V)}\cdot poly(n)$.
Finally, unless 3-SAT can be solved in time $2^{o(n)}$, there is no algorithm
which decides if $G\models \phi$ in time $O(2^{2^{o(k+q_S+q_V)}}\cdot
poly(n))$.

\end{theorem}

\begin{proof}

We have already observed that the construction we described has $k=O(\log n)$,
$q_S=1$ and $q_V=O(poly(\log^* n))$, so $k+q_S+q_V = O(\log n)$. Since the
construction can clearly be performed in polynomial time, if we had an
algorithm to decide if $G\models \phi$ in time $2^{O(k+q_S+q_V)}\cdot poly(n)$
this would imply a polynomial time algorithm for 3-SAT. If we had an algorithm
for the same problem running in time $O(2^{poly(k+q_S+q_V)}\cdot poly(n))$ this
would imply an algorithm for SAT with running time $2^{poly\log(n)}$. Finally,
an algorithm running in time $O(2^{2^{o(k+q_S+q_V)}}\cdot poly(n))$ would imply
an algorithm for SAT running in $2^{o(n)}\cdot poly(n)$. \qed

\end{proof}

\begin{theorem} 

Let $\phi$ be a FO formula with $q_v$ vertex quantifiers and $G$ a graph with
vertex cover $k$. Then, unless FPT=W[1], there is no algorithm which decides if
$G\models \phi$ in time $O(2^{O(k+q_V)}\cdot poly(n))$.

\end{theorem}

\begin{proof}

We use the same construction, but begin our reduction from Weighted 3-SAT,
a well-known W[1]-hard parameterized problem. Suppose we are given a 3-CNF
formula and a number $w$  and we are asked if the formula can be satisfied by
setting exactly $w$ of its variables to true. The formula $\phi$ we construct
is exactly the same, except that we replace the $\exists S$ with $\exists x_1
\exists x_2 \ldots \exists x_w (\bigwedge_{1\le i< j\le w} x_i \neq x_j)$ and
all occurences of $x\in S$ with $\bigvee_{1\le i\le w} x=x_i$.  It is not hard
to see that the informal meaning of $\phi$ now is to ask whether there exists a
set of exactly $w$ distinct variables such that setting them to true makes the
formula true.

We now have $q_V=w+ O(poly(\log^*n))$ so an algorithm running in time
$2^{O(k+q_V)}\cdot poly(n)$ would imply an algorithm for Weighted 3-SAT running
in $2^{O(w)}\cdot poly(n)$, and thus that FPT=W[1]. \qed

\end{proof}

\section{Neighborhood Diversity} \label{sec:nd}

In this Section we give some general results on the new graph parameter we have defined, neighborhood diversity.
We will use $nd(G),tw(G),cw(G)$ and $vc(G)$ to denote the neighborhood diversity, treewidth, cliquewidth and
minimum vertex cover of a graph $G$. We will call a partition of the vertex set of a graph $G$ into $w$ sets such
that all vertices in every set share the same type a neighborhood partition of width $w$.

First, some general results

\begin{theorem} \label{thm:ndgeneral}

\begin{enumerate}

\item Let $V_1,V_2,\ldots, V_w$ be a neighborhood partition of the vertices of a graph $G(V,E)$. Then each $V_i$
induces either a clique or an independent set. Furthermore,  for all $i,j$ the graph either includes all possible
edges from $V_i$ to $V_j$ or none.

\item For every graph $G$ we have $nd(G) \le 2^{vc(G)}+vc(G)$ and $cw(G)\le nd(G) +1$. Furthermore, there exist graphs
of constant treewidth and unbounded neighborhood diversity and vice-versa.

\item There exists an algorithm which runs in polynomial time and given a graph $G(V,E)$ finds a neighborhood partition
of the graph with minimum width.

\end{enumerate}

\end{theorem}

\begin{proof}

For the first statement, to show that every $V_i$ induces either a clique or an independent set, we may assume
that $|V_i|\ge 3$, otherwise the statement is trivial. Suppose that some $V_i$ includes at least one edge $(u,v)$.
Then for every other pair of vertices $w,w'$ we know that $w$ must be connected to $v$ since $w$ and $u$ have the
same type. With a symmetric argument we conclude that all the edges $(w,u),(w,v),(w',u),(w',v)$ must exist in the
graph. Finally, because $w$ and $u$ have the same type and we concluded that $(w',u)$ is an edge, we must have
$(w,w')$ as well. This is true for any pair of vertices $(w,w')$ so if $V_i$ has at least one edge it is a clique.
Another way to see this observation is to say that the property of two vertices having the same type is an
equivalence relation.

For the edges between $V_i$ and $V_j$, suppose that there exists at least an edge $(u,v)$ between them and let
$w\in V_i$, $w'\in V_j$. $v$ has the same type as $w'$, therefore $(u,w')$ must be an edge. Now, $w$ has the same
type as $u$ so $(w,w')$ must also be an edge, and once again this is true for any $w,w'$.

We have already shown the first part of the second statement. For the part with cliquewidth, we remind the reader
that the graphs of cliquewidth $k$ are those which can be constructed by repeated application of the following
operations: introducing a new vertex with a label in $\{1,\ldots,k\}$, joining all vertices of label $i$ with all
vertices of label $j$, renaming all vertices of label $i$ to label $j$ and taking disjoint union of two graphs of
cliquewidth at most $k$. We must show how to construct a graph in such a way starting from a neighborhood
partition of width $w$, using at most $w+1$ labels. The labels in $\{1,\ldots,w\}$ will only be used for the
vertices of the corresponding set in the partition, while the extra label will be used to construct the cliques.
For each $V_i$, if $V_i$ is an independent set introduce $|V_i|$ new vertices with label $i$. If $V_i$ is a clique
repeat $|V_i|$ times: introduce a new vertex of label $w+1$, join all vertices of label $i$ to $w+1$ and rename
$w+1$ to $i$. After all the vertices have been introduced, for all $i,j$ for which the graph had all edges between
$V_i$ and $V_j$ join the vertices labeled $i$ with those labeled $j$.

To see why treewidth is incomparable to neighborhood diversity consider the examples of a complete bipartite graph
$K_{n,n}$ and a path on $n$ vertices.

Finally, let us argue why neighborhood diversity is computable in polynomial time. First, observe that
neighborhood diversity is closed under the taking of induced subgraphs, that is, if $G'(V',E')$ is an induced
subgraph of $G(V,E)$ then $nd(G')\le nd(G)$, because a neighborhood partition of $G$ is also valid for $G'$. We
will work inductively: order the vertices of the input graph $G$ in an arbitrary way and suppose that we have
found an optimal neighborhood partition of the graph induced by the first $k$ vertices into $w$ sets,
$V_1,V_2,\ldots, V_w$. From our observation regarding induced subgraphs we know that the optimal partition of the
graph induced by the first $k+1$ vertices will need at least $w$ sets. Let $u$ be the next vertex. There are two
cases: either $u$ can be placed in some $V_i$ giving us a valid and optimal neighborhood partition of the first
$k+1$ vertices or not, and this can easily be verified in polynomial time. In the second case, there must exist in
each $V_i$ a vertex $v_i$ such that $v_i$ and $u$ have different types. This means that we have a set of $w+1$
vertices which have mutually incompatible types, which implies that the optimal neighborhood partition needs at
least $w+1$ sets. This can be achieved by adding to the partition we have a new singleton set $\{u\}$. \qed

\end{proof}

Taking into account the observations of Theorem \ref{thm:ndgeneral} we
summarize what we know about the graph-theoretic and algorithmic properties of
neighborhood diversity and related measures in Figure \ref{fig:hier}.

\begin{figure}[htb]

\begin{tabular}{cc}

\begin{minipage}{0.3\textwidth}\centering \includegraphics[scale=0.3]{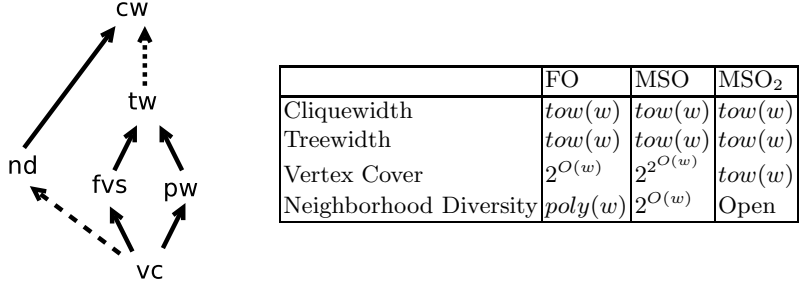}\end{minipage} & 
 
 \begin{minipage}{0.7\textwidth} 
 
\begin{tabular}{|l|l|l|l|} 
\hline 
			& FO			& MSO				& MSO$_2$ \\
\hline
Cliquewidth		& $tow(w)$		&$tow(w)$			&$tow(w)$	\\
Treewidth		&$tow(w)$ 		&$tow(w)$			&$tow(w)$	\\
Vertex Cover		&$2^{O(w)}$		&$2^{2^{O(w)}}$			&$tow(w)$	\\
Neighborhood Diversity	&$poly(w)$		&$2^{O(w)}$			&Open	\\
\hline

\end{tabular}
 
 \end{minipage}

\end{tabular}

\caption{A summary of the relations between neighborhood diversity and other
graph widths. Included are cliquewidth, treewidth, pathwidth, feedback vertex
set and vertex cover. Arrows indicate generalization, for example bounded
vertex cover is a special case of bounded feedback vertex set. Dashed arrows
indicate that the generalization may increase the parameter exponentially, for
example treewidth $w$ implies cliquewidth at most $2^w$. The table summarizes
the best known model checking algorithm's dependence on each width for the
corresponding logic.  } \label{fig:hier}

\end{figure}

There are several interesting points to make here. First, though our work is motivated by a specific goal, beating
the lower bounds that apply to graphs of bounded treewidth by concentrating on a special case, it seems that what
we have achieved is at least somewhat better; we have managed to improve the algorithmic meta-theorems that were
known by focusing on a class which is not necessarily smaller than bounded treewidth, only different. However, our
class is a special case of another known width which generalizes treewidth as
well, namely cliquewidth. Since the lower bound results which apply to
treewidth apply to cliquewidth as well, this work can perhaps be viewed more
appropriately as an improvement on the results of \cite{CourcelleMR00}  for
bounded cliquewidth graphs.

Second, and perhaps more interesting, is the fact that in this paper we have
almost entirely ignored the case of MSO$_2$ logic, focusing entirely on
MSO$_1$. The very interesting hardness results shown in \cite{FominGLS09}
demonstrate that the tractability of MSO$_2$ logic is in a sense the price one
has to pay for the additional generality that cliquewidth provides over
treewidth. Thus, a natural question to ask is whether this is the case with
neighborhood diversity as well; is it true that in the process of generalizing
from vertex cover (where MSO$_2$ is linear-time decidable by Courcelle's
theorem) to neighborhood diversity we have sacrificed MSO$_2$ logic?
Furthermore, it is natural to ask if the currently known results for MSO$_2$
logic can be improved in the same way as we did for MSO$_1$, either for
neighborhood diversity or just for bounded vertex cover.

Though we cannot yet fully answer the above questions related to MSO$_2$, we
can offer some first indications that this direction might merit further
investigation. In \cite{FominGLS09} it is shown that MSO$_2$ model checking is
not fixed-parameter tractable when the input graph's cliquewidth is the
parameter by considering three specific MSO$_2$-expressible problems and
showing that they are W-hard. The problems considered are Hamiltonian cycle,
Graph Chromatic Number and Edge Dominating Set. Even though we will not provide
a general meta-theorem to show that MSO$_2$ logic is tractable for bounded
neighborhood diversity we will show that at least it is impossible to show that
it is intractable by considering these problems. In other words, we will show
how these three problems can be solved efficiently on graphs of small
neighborhood diversity. Since small neighborhood diversity is a special case of
small cliquewidth, where these problems are hard, this result could be of
independent interest.

\begin{theorem}

Given a graph $G$ whose neighborhood diversity is $w$, there exist algorithms running in time $O(f(w)\cdot
poly(|G|))$ that decide Hamiltonian cycle, Graph Chromatic Number and Edge Dominating Set.

\end{theorem}

\begin{proof}

We will make use of an auxiliary graph $G'$ on $w$ vertices. Each vertex of $G'$ corresponds to a set in an
optimal neighborhood partition of $G$ and two vertices of $G'$ have an edge iff the corresponding sets of the
partition of $G$ have all possible edges between them.

First, for the chromatic number. Observe that if a set $V_i$ of a neighborhood partition of $G$ induces an
independent set, we can delete all of its vertices but one, without affecting the graph's chromatic number,
because there always exists an optimal coloring where all the vertices of $V_i$ take the same color. So, we can
assume without loss of generality that all the sets $V_i$ of a neighborhood partition of $G$ induce cliques (some
of them of order one).

Consider now a coloring of the graph $G'$ with the following objective function: for each color $i$ used, its
weight is the size of the largest clique that corresponds to a vertex of $G'$ colored with $i$. The objective is
to minimize the sum of the weights of the colors used. It is not hard to see that this problem can be solved in
time $O(w^w\cdot \log n)$ by checking through all possible colorings of the vertices of $G'$. Also, from such a
coloring of $G'$ we can infer a coloring of $G$ that uses as many colors as the weight of the coloring: for every
color $i$ used in $G'$ create a new set of colors of size equal to the color's weight. This is sufficient to color
all the cliques of $G$ that correspond to vertices of $G'$ colored with $i$.

What remains is to argue why this leads to an optimal coloring. Suppose we have an optimal coloring of $G$ and
order the sets of a neighborhood partition in order of decreasing size, that is, $|V_1|\ge|V_2|\ge \ldots \ge
|V_w|$. We will say that $V_i$ and $V_j$ have ``similar'' colors in this optimal coloring of $G$ when there is a
color that appears in both $V_i$ and $V_j$. From the coloring of $G$ we infer a coloring of $G'$ as follows: while
there are still uncolored vertices of $G'$, take the first set of the partition of $G$ (in order of size) that
corresponds to a still uncolored vertex of $G'$. Use a new color for its corresponding vertex in $G'$ and also for
all the vertices that correspond to sets with colors similar to it.

When we are done, we will have a proper coloring of $G'$, because if two sets $V_i$,$V_j$ are joined by an edge
they cannot have similar colors. Furthermore, the weight of the coloring of $G'$ we obtain is a lower bound on the
number of colors used in the original coloring of $G$ we assumed. This is because when we pick a set $V_i$ and use
it to introduce a new color we know that it does not have similar colors with any of the sets we have picked so
far. Because all the sets picked induce cliques and do not have similar colors (i.e. no color is reused) we know
that the original coloring of $G$ uses at least as many colors as the sum of the sizes of the sets picked. Thus,
if our algorithm found that the optimal solution to the weighted coloring problem for $G'$ has weight $w$, this
means that $w$ colors are needed to color $G$, because a coloring of $G$ with $w-1$ colors would give a solution
to the coloring problem of $G'$ with weight at most $w-1$.

For the Hamiltonian cycle problem, we will once again use the graph $G'$. We define the weight of every vertex of
$G'$ to be the size of its corresponding set in the neighborhood partition of $G$. Now, the problem of finding a
Hamiltonian cycle in $G$ can be reduced to the problem of finding a closed walk of $G'$, such that every vertex
that corresponds to an independent set is visited a number of times exactly equal to its weight, while every
vertex corresponding to a clique is visited at least once and at most as many times as its weight.

This problem of looking for a walk on $G'$ can be solved in time $O(n^{w^{2}})$. Replace each edge with two
directed arcs of opposite direction. Now, for each of the at most $w^2$ arcs, we must decide how many times it
will be used, a value upper-bounded by $n$. If we have decided on such values for all arcs we can easily check if
a walk with the desired properties can be made from them. Replace each arc with a number of parallel arcs of the
same direction equal to the value decided for it. Now, we can obtain a walk if the resulting multi-graph is
Eulerian (that is, all vertices have the same in-degree as out-degree) and also the in-degrees of the vertices
follow the conditions we have stated for the number of times the vertex must be visited.

In order to improve this to an FPT algorithm, we rely on an old but seminal
result by Lenstra \cite{lenstra1983integer} (later further improved by Kannan
\cite{kannan1987minkowski}), which states that the feasibility of an ILP
programs of size $n$ with $k$ variables can be solved in time $f(k)\cdot
poly(n)$, i.e.  bounded-variable ILP is FPT.  This is a result that has
attracted considerable interest in the parameterized complexity community and
it has long been a topic of interest to find examples of its application.  Here
we observe that in the above algorithm we are trying to decide on values for
$w^2$ variables. For each variable the constraints can easily be expressed as
linear inequalities: for each vertex we have to make sure that the in-degree is
equal to the out-degree and also that the in-degree falls in a specified
interval. Therefore, by expressing our problem as a system of linear
inequalities we obtain an FPT algorithm.

Finally, in the edge dominating set problem, we are asked to find a set of edges of minimum size such that all
other edges share an endpoint with one of the edges we selected. This problem is equivalent to the minimum maximal
matching problem, where we are trying to find a minimum size independent set of edges that cannot be extended by
picking another edge of the graph. To see why the optimal solution to the edge dominating set problem is always a
matching, suppose that we have a solution $S$ which includes two edges $(u,v),(u,v')$. Now, if all the neighbors
of $v'$ are incident on an edge of $S$ we can simply remove $(u,v')$ from $S$ and improve the size of the
solution. If there is a neighbor $w$ of $v'$ that is not incident on an edge of $S$ we can replace $(u,v')$ with
$(w,v')$ in $S$. To see why a solution to the edge dominating set problem is a maximal matching, suppose that it
was not. Then there would be two unmatched vertices connected by an edge, which would imply that this edge is not
dominated.

Our algorithm will proceed as follows: for every minimal vertex cover $V'$ of $G'$ repeat the following (there are
at most $2^w$ vertex covers to be considered): from $V'$ infer a vertex cover of $G$ by placing into the vertex
cover all the vertices that belong in a type whose corresponding vertex is in $V'$. Also place in the vertex cover
all but one (arbitrarily chosen) vertex of every vertex type that induces a clique but whose corresponding vertex
is not in $V'$. Call the resulting vertex cover of $G$ $V''$. Find a maximum matching on the graph induced by
$V''$, call it $M_1$. Take the bipartite graph induced by the unmatched vertices of $V''$ and $V\setminus V''$ and
find a maximum matching there, call it $M_2$. The solution produced is $M_1\cup M_2$. After repeating this for all
vertex covers of $G'$, pick the smallest solution.

Now we need to argue why this solution is optimal. Let $S$ be an optimal solution for $G$. We say that a set of
the neighborhood partition $V_i$ is full if all of its vertices are incident on edges of $S$. If we take in $G'$
the corresponding vertices of the full sets of $G$, they must form a vertex cover of $G'$, otherwise there would
be two neighboring vertices with neither having any edge of $S$ incident to it, which would mean that $S$ is not
maximal. This is a vertex cover of $G'$ considered by our algorithm, since our algorithm considers all vertex
covers of $G'$, call it $V'$. Let $V''$ be again the vertex cover of $G$ our algorithm derived from $V'$ by also
including a minimal number of vertices from each remaining clique. Let $V^*$ be the set of vertices of $G$
incident on some edge of $S$, which must also be a vertex cover of $G$. Without loss of generality we will assume
that $V''\subseteq V^*$, because the two vertex covers of $G$ agree on taking all vertices of the full sets and
$V''$ takes a minimal number of vertices from every other clique. Even if $V^*$ leaves out a different vertex from
some clique because all the vertices of the clique have the same neighbors we can apply an exchanging argument and
transform $S$ appropriately without increasing its size so that both sets leave out the same vertex.

Now note that $|M_2|\le |V''|-2|M_1|$. So our algorithm's solution has size at most $|V''|-|M_1|$. On the other
hand the optimal solution $S$ includes some edges with both endpoints in $V''$, call this set $S_1$. Because $M_1$
is a maximum matching, $|S_1|\le |M_1|$. From what we have so far, the fact that all vertices of $V^*$ are matched
by $S$ and the fact that $V''$ is a vertex cover, so $V^*\setminus V''$ induces no edges we have $|V^*| = |V^*\cap
V''| + |V^*\setminus V''| = |V''| + |V''|-2|S_1| \ge 2|V''| -2 |M_1|$. This implies that $|S|\ge |V''|-|M_1|$
which concludes the proof. \qed

\end{proof}

\section{Conclusions and Open Problems}

The vast majority of treewidth-based algorithmic results, including Courcelle's
theorem, rely on the exploitation of small graph separators. The limit of this
technique is that in the worst case its complexity can be a tower of
exponentials, depending on the problem at hand. In this paper we have exploited
a different technique which groups vertices into equivalence classes, depending
on their neighborhoods. Using this we were able to offer a huge improvement on
the currently known meta-theorems for MSO and FO tractability for the special
case of graphs of bounded vertex cover, and we also showed that our
meta-theorems are in some sense ``optimal''. In the process we defined a new
graph complexity metric which measures how well our technique can be applied on
a given graph.

One direction for future research now is the further investigation of the
properties of neighborhood diversity. From the results of this paper we know
that small neighborhood diversity implies tractability for FO and
MSO-expressible problems and we also know that neighborhood diversity can be
solved optimally in polynomial time (a rarity in the realm of graph widths!).
The main theoretical problem left open is whether MSO$_2$ logic is tractable
for small neighborhood diversity. The main practical problem on the other hand
is whether graphs of small neighborhood diversity do appear often in common
applications. It is worth remembering that treewidth is a successful complexity
measure not only because many problems are solvable for graphs of small
treewidth but also because empirically many practical instances seem to have
small treewidth. Is the same true for neighborhood diversity? Any evidence
pointing to a positive answer to this question would greatly motivate further
research on the topic.

Other directions to consider along the lines of this paper are, first, trying
to achieve results similar to this paper's for other restrictions of treewidth.
The most notable case here is probably graphs of bounded max leaf number,
another problem posed explicitly by Fellows. Second, another interesting next
step would be to attempt to prove tractability (or intractability) for logics
larger than MSO$_2$ for bounded vertex cover, for example for a logic that
includes the ability to quantify over orderings of the vertex set. To this end,
the results of \cite{FellowsLMRS08} give some positive indication that this may
be possible.

\bibliographystyle{plain}
\bibliography{metatheorems}

\end{document}